\documentclass[12pt,reqno]{amsart}
\usepackage{amsmath,amsthm,amssymb,amscd,url,rotating}
\usepackage[left=4cm,right=4cm,top=4cm,bottom=4cm]{geometry}
\usepackage{graphicx,epsf,epsfig,exscale}
\usepackage[usenames,dvipsnames]{color}
\usepackage{hyperref} 
\hypersetup{colorlinks=true}   
\usepackage[all]{xy}
\xyoption{curve}

\allowdisplaybreaks

\newcommand{\TITLE}{HOW TO MAKE THE MOST OF A SHARED MEAL: \\ PLAN THE LAST BITE FIRST.}
\newcommand{\TITLERUNNING}{HOW TO MAKE THE MOST OF A SHARED MEAL}
\newcommand{\DATE}{September 23, 2011}

\theoremstyle{plain} 
\newtheorem{theorem}{Theorem}

\newtheorem{lemma}[theorem]{Lemma}

\theoremstyle{definition}

\theoremstyle{remark}

\newtheorem*{acknowledgement}{Acknowledgements}

  {\end{list}}

  {\end{list}}

\newcommand{\RR}{\mathbb{R}}

\newcommand{\floor}[1]{\left\lfloor {#1} \right\rfloor}
\newcommand{\ceiling}[1]{\left\lceil {#1} \right\rceil}

\title[\TITLERUNNING]{\TITLE}
\date{\DATE}
\author{Lionel Levine}
\address{Department of Mathematics, MIT, 77 Massachusetts Ave., Cambridge, MA 02139. 
{\tt \url{http://math.mit.edu/~levine}}}
\email{levine@math.mit.edu}
\author{Katherine E. Stange} 
\address{Department of Mathematics, Stanford University, 450 Serra Mall, Building 380, Stanford, CA 94305.
{\tt \url{http://math.katestange.net}}}
\email{stange@math.stanford.edu}

\begin{document}

\begin{abstract}

If you are sharing a meal with a companion, how best to make sure you get your favourite mouthfuls?   Ethiopian Dinner is a game in which two players take turns eating morsels from a common plate.  Each morsel comes with a pair of utility values measuring its tastiness to the two players.  
Kohler and Chandrasekaran discovered a good strategy --- a subgame perfect equilibrium, to be exact --- for this game.  We give a new visual proof of their
result. The players arrive at the equilibrium by figuring out their last move first and working backward. We conclude that it's never too early to start thinking about dessert.
\end{abstract}

\maketitle


\section*{Introduction}
\label{section:introduction}

Consider two friendly but famished acquaintances sitting down to dinner at an Ethiopian restaurant.  The food arrives on a common platter, and each friend has his own favourite and not-so-favourite dishes among the spread.  Hunger is a cruel master, and each of our otherwise considerate companions finds himself racing to swallow his favourites before his comrade can scoop them up.  Each is determined to maximize his own gastronomic pleasures, and could not care less about the consequences for his  companion.

An \emph{Ethiopian Dinner} is a finite set
\[
D = \{ m_1, \ldots, m_n \}, \quad \mbox{ with } \quad  m_i=(a_i,b_i)
\]
whose elements are called \emph{morsels}.  Each morsel $m_i$ is an ordered pair of real numbers $(a_i,b_i)$.  Two players, Alice and Bob, take turns removing one morsel from~$D$ and eating it.  The morsels are discrete and indivisible: each morsel can be eaten exactly once, and the game ends when all morsels have been eaten.  The larger the value of $a_i$, the tastier the morsel for Alice; the larger $b_i$, the tastier for Bob.  We assume that the \emph{players know one another's preferences}, and that the preferences are totally ordered, that is, $a_i \neq a_j$ and $b_i \neq b_j$ for $i\neq j$.  Alice's score is the sum of the~$a_i$ for the morsels~$m_i$ she eats, while Bob's score is the sum of the~$b_i$ for the morsels~$m_i$ he eats. 

In such a game, the players are not adversaries; in fact, the game may end quite peaceably and successfully for both players if they have dissimilar tastes.  The question we are interested in is this:  if a player acts rationally to maximize her own score, and assumes that her meal partner does the same, what should be her strategy?  

Eating your favourite morsel on the first move of an Ethiopian Dinner is not necessarily a good strategy.  For example, if the dinner is
	\[ D = \{ (1,2), (2,3), (3,1) \}, \]
then Alice's favourite morsel is $(3,1)$.  If she takes this morsel first, then Bob will take $(2,3)$, leaving Alice with $(1,2)$ for a total score of $4$.  Instead Alice should snag $(2,3)$ on the first move; after Bob takes $(1,2)$, Alice can finish up with $(3,1)$ for dessert and a total score of $5$.

If deciding on the first move in an Ethiopian Dinner appears complicated, the \emph{last move} is a different matter.  The subject of this paper is a strategy discovered by Kohler and Chandrasekaran~\cite{KC}, which we call the \emph{crossout strategy}.  Its mantra is:
\begin{quote}
  ``{\em Eat your opponent's least favourite morsel on your own last move.}''
\end{quote}
To arrive at this strategy, each player reasons informally as follows. My opponent will never choose her least favourite morsel, unless it is the only one left; therefore, unless this is my last move, I can safely save my opponent's least favourite morsel for later.  

This reasoning predicts that if, say, Bob has the last move of the game, then Bob's last move will be to eat Alice's least favourite morsel.  Because this is a game of perfect information, both players can use this reasoning to predict with certainty the game's last move.  We now cross out Alice's least favourite morsel from the dinner $D$ to arrive at a smaller dinner $D'$ in which Alice has the last move.  The same reasoning now implies that on her last move, Alice will eat Bob's least favourite morsel in $D'$.  We then cross out Bob's least favourite morsel from $D'$ and proceed inductively, alternately crossing out Alice's least favourite and Bob's least favourite among the remaining morsels until all morsels have been crossed out.  The \emph{crossout strategy} is to eat the \emph{last} morsel to be crossed out. 

\subsection*{Strategies and equilibria}

To convert the informal reasoning above into a proof that crossout is a ``good strategy,'' we need to define both words!  A \emph{strategy} is any function $s$ from nonempty dinners to morsels, such that $s(D) \in D$ for all nonempty dinners $D$.  Playing strategy $s$ means that if it is your turn and the remaining set of morsels is $D$, then you eat the morsel $s(D)$.  For example, the \emph{Alice greedy} strategy chooses the morsel $m_i \in D$ such that $a_i$ is largest.  Either player could play this strategy. (It would be pretty spiteful of Bob to do so!)  Other examples are the ``competitive'' strategy that chooses the morsel maximizing $a_i+b_i$, the ``Alice cooperative'' strategy choosing the morsel maximizing $a_i - b_i$, and the ``Alice masochistic'' strategy choosing the morsel minimizing $a_i$.  A strategy can be as complicated as you like: perhaps if the number of morsels is prime, you greedily choose your own favourite morsel, and otherwise you spitefully choose your opponent's favourite. 

The appropriate notion of good strategy depends on the class of games one is considering.  Ethiopian Dinner is a nonzero-sum game: one player's gain may not be the other's loss.  In such a game, the basic requirement of any pair of good strategies (one for Alice, one for Bob) is that they form a \emph{Nash equilibrium}, which means that neither player can benefit himself by changing strategies unilaterally.  

A Nash equilibrium represents a stable, predictable outcome: Alice can declare, ``I am playing my equilibrium strategy, and you'd do best to play yours.''  If Bob responds rationally by playing his own  equilibrium strategy, then both players know how the game will turn out.  

\begin{figure}
  \begin{center}
\[
\begin{xy}
0;
(0.95,0):
(10,10)*{\includegraphics[width=4in]{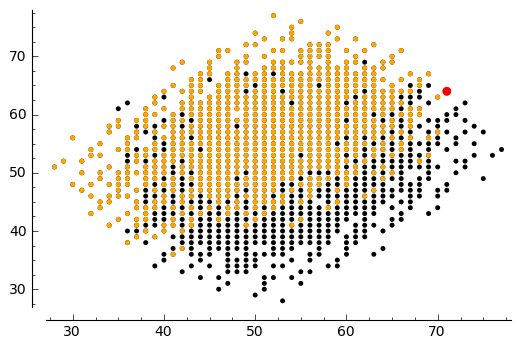}};
(-46,-5)*{\begin{sideways}\begin{footnotesize}$\mbox{Bob's score}\longrightarrow$\end{footnotesize}\end{sideways}};
(-20,-28)*{\mbox{\begin{footnotesize}Alice's score\end{footnotesize}}\longrightarrow};
\end{xy}
\]
  \end{center}
  \caption{Plot of the score pairs for all possible outcomes of a permutation dinner $D$ of size $14$.  The large dot {\Large ${\color{red}\bullet}$} at upper right represents the score Alice and Bob receive if they both play the crossout strategy.  The light dots ${\color{BurntOrange}\bullet}$ represent scores for strategy pairs of the form $(s,c)$ for $s$ arbitrary: these are all of the outcomes Alice can obtain playing against Bob's crossout strategy.  According to Theorem~\ref{thm:crossoutequil}, among these outcomes she does best when she herself plays crossout.  Dark dots $\bullet$ represent the outcomes of all other strategy pairs.  Produced using Sage Mathematics Software~\cite{Sage}.}
  \label{fig:typicalscores}
\end{figure}

\subsection*{Subgame perfect equilibria}

A game may have many equilibrium strategy pairs, some with better outcomes than others, so one tends to look for equilibria with further desirable properties.  Which properties again depends on the class of games being considered.  In game theory lingo, Ethiopian Dinner is a \emph{perfect-information} \emph{non-cooperative game in extensive form}.  That is, both players know the values $a_1,\ldots,a_n$ and $b_1,\ldots,b_n$ (perfect information); the players may not bargain or make side deals (non-cooperative); and the players alternate making moves (extensive form).  


Non-cooperative games model situations in which the players have no way of communicating (our dinner guests may be too busy stuffing their faces to bother with conversation) or are forbidden to collude.  For instance, airlines are forbidden by law from colluding to fix prices.  Colluding to fix the outcome of a meal is still legal in most countries, but Alice might nevertheless be dissuaded by cultural taboo from making propositions like ``If you pay me fifty cents I promise not to eat any more spinach.''  

An Ethiopian Dinner with $n$ morsels is certain to end in $n$ moves.  A widely accepted notion of a good strategy for games of this type (perfect information, non-cooperative, extensive form, finite length) is the \emph{subgame perfect equilibrium}.   This is a refinement of the Nash equilibrium which requires that the strategies remain in equilibrium when restricted to any subgame.  In our case, a subgame is just a \emph{subdinner} consisting of a subset of the morsels, with the same player moving last.  A subgame perfect equilibrium is robust in the sense that even if one player, say Bob, makes a ``mistake'' on a particular move by deviating from his equilibrium strategy, Alice can confidently continue playing her equilibrium strategy because the same strategy pair is still an equilibrium of the resulting subgame.  See, e.g.,~\cite{Owen} and \cite{Rasm} for background on these concepts.

Let $c$ be the crossout strategy described above for Ethiopian Dinner.  We will give a new proof of the following theorem, which is due to Kohler and Chandrasekaran~\cite{KC}.

\begin{theorem}
  \label{thm:crossoutequil}
  The pair $(c,c)$ is a subgame perfect equilibrium. 
\end{theorem}

\noindent In other words, if Alice plays crossout, then Bob cannot benefit himself by playing a different strategy, and vice versa.

Figure~\ref{fig:typicalscores} illustrates Theorem~\ref{thm:crossoutequil} in the case of a particular \emph{permutation dinner}, that is, a dinner of the form
\[
D = \{ (1,b_1), (2,b_2), \ldots, (n,b_n) \}
\]
where $b_1,\ldots,b_n$ is a permutation of the numbers $1,\ldots,n$.  Each dot in the figure represents the outcome of a strategy pair, with Alice's score plotted on the horizontal axis and Bob's score on the vertical axis, for the following permutation dinner of size $14$: 
\begin{align*}
D = &\{  (1, 6), (2, 14), (3, 10), (4, 3), (5, 7), (6, 5), (7, 9), (8, 8), \\ &\quad\quad (9,4), (10, 13), (11, 12), (12, 11), (13, 2), (14, 1) \}.
\end{align*}
To visualize Theorem~\ref{thm:crossoutequil}, note that the large dot {\Large ${\color{red}\bullet}$} in Figure~\ref{fig:typicalscores}, which represents the outcome when both players play crossout, is rightmost among all possible outcomes achievable by Alice given that Bob plays crossout (such outcomes are indicated by light dots {${\color{BurntOrange}\bullet}$}).

Recall that \emph{any} function $s$ from dinners to morsels such that $s(D) \in D$ for all dinners $D$ counts as a strategy.  In particular, let $k$ be the number of moves that Alice makes in $D$.  For any subset $A \subset D$ of size $k$, there is a strategy pair $(s,s')$ such that if Alice plays $s$ and Bob plays $s'$, then Alice will eat precisely the morsels in $A$ and Bob will eat the rest.  This is how we computed the set of all dots in Figure~\ref{fig:typicalscores}.  
 
Part of the difficulty of proving that a pair of strategies is in equilibrium -- and also part of the power of the result once it is proved -- stems from the vast number of possible strategies.  How are we to rule out that \emph{any conceivable} strategy $s$, which might be arbitrarily subtle, clever and complicated, performs better against $c$ than does $c$ itself?  
For all we know, Alice has access to unlimited computational resources, and her strategy $s$ could involve factoring enormous integers or solving large instances of the traveling salesman problem.  Nevertheless, Theorem~\ref{thm:crossoutequil} ensures that if Bob plays $c$, then Alice would do at least as well playing $c$ as $s$. 
In a sense, a theorem about equilibrium is a theorem about the limits of intelligence: 
An equilibrium enables you to hold the line against a smarter opponent!  
 


\subsection*{Crossout is an efficiently computable equilibrium}

In games arising in the real world, for instance in evolutionary dynamics and in economics, the appeal of the Nash equilibrium concept is twofold.  First, it can explain why we observe certain strategies and not others.  Second, even in the case of a game that has multiple equilibria and lacks a well-defined ``best'' outcome, knowing an explicit equilibrium provides certainty.  Alice simply announces her intention to play crossout, refers Bob to the proof of Theorem~\ref{thm:crossoutequil} and trusts that his own best interest compels him to follow suit.  What might have been a tense evening with an unpredictable outcome becomes a more relaxed affair in which each player can predict in advance which morsels she will be gobbling up.

To reap these benefits, the players must be able to compute an equilibrium pair, not just know that one exists!  A recent strand of research, popularized by the slogan ``\emph{if your laptop can't find it, then, probably, neither can the market},'' has explored the tendency for equilibria to be extremely difficult to compute~\cite{DGP}.  The general existence proof for subgame perfect equilibria \cite[VIII.2.10]{Owen} uses a backward induction from the last move: if converted naively into an algorithm, it would seem to require searching through exponentially many move sequences in order to find an equilibrium.  This kind of brute force search is typically out of the question even for games of moderate size (for example, an Ethiopian Dinner of $n$ morsels has $n!$ possible move sequences).
For this reason, it is always interesting to identify special classes of games that have efficiently computable equilibria.   The crossout equilibrium for Ethiopian Dinner is an example: if both players play the crossout strategy, then they eat the morsels in reverse order of the crossouts.  In this case, the entire move sequence of the dinner can be worked out in the order $n \log n$ time it takes to sort the two lists $a_1,\ldots,a_n$ and $b_1,\ldots,b_n$. 

\section*{Proof of equilibrium}

\subsection*{Dinners and strategies}
\label{section:dinnersstrategies}

A \emph{dinner} is a finite set of morsels \[D = \{m_1,\ldots,m_n\}.\]  Each morsel $m \in D$ comes with a pair of real numbers $u_A(m), u_B(m)$ representing its utility to Alice and Bob.   We often write $m$ as an ordered pair,
\[
m = \left( u_A(m), u_B(m) \right).
\]

We adopt the convention that \emph{Bob has the final move} by default.  Since moves alternate, the first move is determined by the parity of $n$. Alice has the first move if $n$ is even, and Bob has the first move if $n$ is odd.  

A \emph{strategy} is a map assigning to any non-empty dinner $D$ a morsel $s(D) \in D$ to be eaten by the first player.  Suppose that $P \in \{ \mbox{Alice, Bob} \}$ is a player, $D$ is a dinner, and that $P$ plays strategy $s$.  If it is $P$'s turn to move, he selects morsel $s(D)$ and receives payoff $u_P(s(D))$.  The remaining dinner is $D - s(D)$, with his opponent to move.  Suppose his opponent plays strategy~$t$.  The \emph{score} $v_P^D(s,t)$ of player $P$ is defined by the recurrence
	\begin{equation} \label{eq:scoredef}
	v_P^D(s,t) = \begin{cases}
	   v_P^{D-s(D)}(s,t) + u_P(s(D))&  \mbox{if $P$ plays first in $D$,} \\
	   v_P^{D-t(D)}(s,t) & \mbox{if $P$ plays second in $D$,} \\
	   0 & \mbox{if }D  = \emptyset. 
	 \end{cases}
	\end{equation}
	where for $m \in D$, the dinner $D - m$ denotes $D$ with morsel $m$ removed.  Since $D$ has finitely many morsels, equation \eqref{eq:scoredef} defines $v_P^D(s,t)$ uniquely. 
	
	Our convention in denoting a player's score is that his own strategy is always the first listed in the ordered pair.  
	
 Formally, we can regard Ethiopian Dinner as a single game whose positions comprise all finite dinners.  A pair of strategies $(s,t)$ is a \emph{subgame perfect equilibrium} for this game if
		\[ v^D_A(s',t) \leq v^D_A(s,t)  \quad \mbox{ and } \quad v^D_B(t',s) \leq v^D_B(t,s) \]
	for all strategies $s'$ and $t'$ and all finite dinners $D$.
	
\subsection*{The crossout strategy}
\label{section:crossoutstrategy}

After giving the formal definition of the crossout strategy described in the introduction, we explain how to visualize it using a ``crossout board'' and prove the lemma that lies at the heart of our argument, the Crossout Board Lemma (Lemma \ref{lem:nmoveout}).

Let $D$ be a set of $n$ morsels.  Write $\ell_A(D)$ for Alice's least favourite morsel in $D$, and $\ell_B(D)$ for Bob's least favourite morsel in $D$.  Let $D_1 = D$, and
	\[ D_{i+1} = D_i - m_i, \qquad i=1,\ldots,n-1 \]
where
	\[ m_i = \begin{cases} \ell_A(D_i), & i \mbox{ odd} \\
					\ell_B(D_i), & i \mbox{ even}. 
			\end{cases} \]
The sequence of morsels
	\[ m_1, m_2, \ldots, m_n \]
is called the \emph{crossout sequence} of $D$.  Note that 
	\begin{align*} 
	&\mbox{$m_1$ is Alice's least favourite morsel in $D$} \\
	&\mbox{$m_2$ is Bob's least favourite morsel in $D-m_1$}\\
	&\mbox{$m_3$ is Alice's least favourite morsel in $D-m_1-m_2$}\\
	&\mbox{$m_4$ is Bob's least favourite morsel in $D-m_1-m_2-m_3$}\\
	&\quad \vdots
	\end{align*}
Now suppose $D$ is a dinner (i.e., a set of $n$ morsels with Bob distinguished to move last).  The \emph{crossout strategy} $c$ is defined by
	$ c(D) = m_n. $
Note that if both players play the crossout strategy, then they eat the morsels in \emph{reverse order} of the crossout sequence:
	\begin{align*}  & m_n = c(D) \\  &m_{n-1} = c(D-m_n) \\ & m_{n-2} = c(D-m_n-m_{n-1}) \\ & \quad \vdots \\
	 & m_1 = c(D-m_n-\cdots-m_{2}).
	\end{align*}
Thus $m_1$, which is Alice's least favourite morsel in $D$, is eaten by Bob on the last turn.
\subsection*{Crossout boards}
\label{section:crossoutboards}

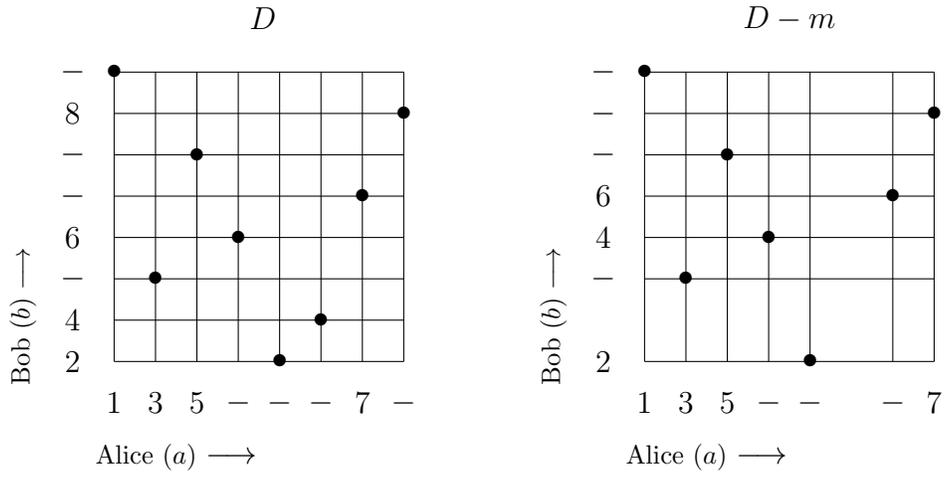
\begin{figure}
\[
\begin{xy}
0;
(0.55,0):
(0,0)*{}; (70,0)*{} **\dir{-};
(0,10)*{}; (70,10)*{} **\dir{-};
(0,20)*{}; (70,20)*{} **\dir{-};
(0,30)*{}; (70,30)*{} **\dir{-};
(0,40)*{}; (70,40)*{} **\dir{-};
(0,50)*{}; (70,50)*{} **\dir{-};
(0,60)*{}; (70,60)*{} **\dir{-};
(0,70)*{}; (70,70)*{} **\dir{-};
(0,0)*{}; (0,70)*{} **\dir{-};
(10,0)*{}; (10,70)*{} **\dir{-};
(20,0)*{}; (20,70)*{} **\dir{-};
(30,0)*{}; (30,70)*{} **\dir{-};
(40,0)*{}; (40,70)*{} **\dir{-};
(50,0)*{}; (50,70)*{} **\dir{-};
(60,0)*{}; (60,70)*{} **\dir{-};
(70,0)*{}; (70,70)*{} **\dir{-};
(30,30)*{\bullet};
(10,20)*{\bullet};
(60,40)*{\bullet};
(50,10)*{\bullet};
(20,50)*{\bullet};
(70,60)*{\bullet};
(0,70)*{\bullet};
(40,0)*{\bullet};
(-10,0)*{2};
(-10,10)*{4};
(-10,20)*{-};
(-10,30)*{6};
(-10,40)*{-};
(-10,50)*{-};
(-10,60)*{8};
(-10,70)*{-};
(0,-10)*{1};
(10,-10)*{3};
(20,-10)*{5};
(30,-10)*{-};
(40,-10)*{-};
(50,-10)*{-};
(60,-10)*{7};
(70,-10)*{-};
(-22,11)*{\begin{sideways}\begin{footnotesize}$\mbox{Bob ($b$)}\longrightarrow$\end{footnotesize}\end{sideways}};
(15,-23)*{\mbox{\begin{footnotesize}Alice ($a$)\end{footnotesize}}\longrightarrow};
(36,83)*{D}
\end{xy}
\quad\quad\quad\quad
\begin{xy}
0;
(0.55,0):
(0,0)*{}; (70,0)*{} **\dir{-};
(0,20)*{}; (70,20)*{} **\dir{-};
(0,30)*{}; (70,30)*{} **\dir{-};
(0,40)*{}; (70,40)*{} **\dir{-};
(0,50)*{}; (70,50)*{} **\dir{-};
(0,60)*{}; (70,60)*{} **\dir{-};
(0,70)*{}; (70,70)*{} **\dir{-};
(0,0)*{}; (0,70)*{} **\dir{-};
(10,0)*{}; (10,70)*{} **\dir{-};
(20,0)*{}; (20,70)*{} **\dir{-};
(30,0)*{}; (30,70)*{} **\dir{-};
(40,0)*{}; (40,70)*{} **\dir{-};
(60,0)*{}; (60,70)*{} **\dir{-};
(70,0)*{}; (70,70)*{} **\dir{-};
(30,30)*{\bullet};
(10,20)*{\bullet};
(60,40)*{\bullet};
(20,50)*{\bullet};
(70,60)*{\bullet};
(0,70)*{\bullet};
(40,0)*{\bullet};
(-10,0)*{2};
(-10,10)*{};
(-10,20)*{-};
(-10,30)*{4};
(-10,40)*{6};
(-10,50)*{-};
(-10,60)*{-};
(-10,70)*{-};
(0,-10)*{1};
(10,-10)*{3};
(20,-10)*{5};
(30,-10)*{-};
(40,-10)*{-};
(50,-10)*{};
(60,-10)*{-};
(70,-10)*{7};
(-22,11)*{\begin{sideways}\begin{footnotesize}$\mbox{Bob ($b$)}\longrightarrow$\end{footnotesize}\end{sideways}};
(15,-23)*{\mbox{\begin{footnotesize}Alice ($a$)\end{footnotesize}}\longrightarrow};
(35,83)*{D-m}
\end{xy}
\]

 \caption{Left: Example of a crossout board for a dinner $D$ with 8 morsels.  Labels on the axes indicate the crossout sequence.  Right: The crossout board for the dinner $D-m$, in which a morsel $m$ has been removed.  By the Crossout Board Lemma, each label on the right is at least as far from the origin as the corresponding label on the left. }
  \label{figure:crossoutboard}
\end{figure}

To prepare for the proof of Theorem \ref{thm:crossoutequil}, it is convenient to illustrate the crossout sequence with a \emph{crossout board}, as in Figure~\ref{figure:crossoutboard}.  We display the dinner on a Cartesian coordinate plane. Each morsel $m=(a,b)$ is graphed as a dot at coordinate $(a,b)$.  Since we assume that the players' preferences are totally ordered, each vertical or horizontal line passes through at most one morsel.  The crossout sequence itself is indicated by writing the number (or \emph{label}) $i$ on the $a$-axis below $m_i$ if $i$ is odd, and on the $b$-axis to the left of $m_i$ if $i$ is even.  

Figure~\ref{figure:crossoutboard} shows the crossout board of the dinner 
	\[ D = \{ (1,8), (2,3), (3,6), (4,4), (5,1), (6,2), (7,5), (8,7) \} \]
and of $D-m$, where $m$ is the morsel $(6,2)$.  It is helpful to imagine placing the labels on a crossout board one at a time in increasing order.  Alice starts at the left and scans rightward, placing the label $1$ below her least favourite morsel.  Then Bob starts at the bottom and scans upward, placing the label $2$ to the left of his least favourite unlabeled morsel.  The players alternate in this fashion until all morsels are labeled.  Note that the labels on each axis appear in increasing order moving away from the origin.  
Alice always performs the first crossout, because of our convention that Bob has the last move.  Hence, the odd labels appear on Alice's axis and the even labels on Bob's axis.


The central lemma needed to show that crossout is an equilibrium is the following.

\begin{lemma}[Crossout Board Lemma]
\label{lem:nmoveout}
Let $D$ be a dinner, and $\widehat{D} \subset D$ a subdinner.  For each $k=1,\ldots,|\widehat{D}|$ the location of label $k$ in the crossout board of $\widehat{D}$ is at least as far from the origin as the location of label $k$ in the crossout board of $D$.
\end{lemma}

\begin{proof}
  Let $B$ be the crossout board for $D$, with crossout sequence
  \[
  m_1, m_2, \ldots, m_{|D|}.
  \]
  Let $\widehat{B}$ be the crossout board for $\widehat{D}$, with crossout sequence
  \[
  \widehat{m}_1, \widehat{m}_2, \ldots, \widehat{m}_{|\widehat{D}|}.
  \]
 For morsels $p$ and $q$ of $D$, we write
  $
  p \mathop{<}_D q
  $
  to mean that $p$ appears before $q$ in the crossout sequence for $D$.  If $p$ and $q$ are also morsels of $\widehat{D}$, then we write
  $
  p \mathop{<}_{\widehat{D}} q
  $
  to mean that $p$ appears before $q$ in the crossout sequence for $\widehat{D}$.
In particular, for any $1 \leq j,k \leq |\widehat{D}|$ we have
  \begin{equation}
      m_j \mathop{<}_D m_k \iff j < k \iff \widehat{m}_j \mathop{<}_{\widehat{D}} \widehat{m}_k. 
    \label{eqn:jlek}
  \end{equation}

Given $1 \leq k \leq |\widehat{D}|$, let us say $k$ is \emph{jumpy} if the label $k$ is strictly closer to the origin in $\widehat{B}$ than in $B$.  We will show that there are no jumpy labels.

Let $P$ be the player who places the label $k$ (so $P$ is Alice if $k$ is odd, Bob if $k$ is even).  When $P$ places the label $k$ on board $B$ next to the morsel $m_k$, this morsel is the closest available to the origin along $P$'s axis.  If $k$ is jumpy, then the morsel $\widehat{m}_k$ is closer to the origin along $P$'s axis, which means that $\widehat{m}_k$ is \emph{unavailable}, that is, it was already labeled in $B$ by some $j<k$ (Figure~\ref{figure:proofaid}).  Hence
  \begin{equation}
    k\mbox{ is jumpy} \implies \widehat{m}_k \mathop{<}_D m_k. 
  \label{eqn:kjumpy1}
\end{equation}

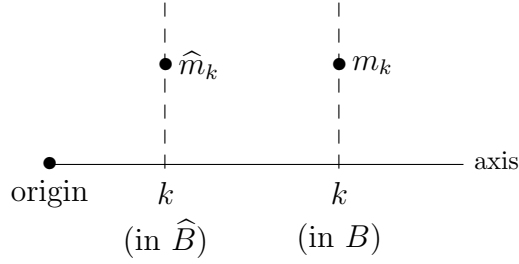
\begin{figure}
\[
\begin{xy}
0;
(0.55,0):
(0,-4)*{}; (100,-4)*{} **\dir{-};
(28,-5)*{}; (28,35)*{} **\dir{--};
(70,-5)*{}; (70,35)*{} **\dir{--};
(28,-11)*{k};
(28,-22)*{ \mbox{(in $\widehat{B}$)}};
(70,-11)*{k};
(70,-22)*{ \mbox{(in $B$)}};
(28,20)*{\bullet};
(70,20)*{\bullet};
(36,20)*{\widehat{m}_k};
(78,20)*{m_k};
(0,-4)*{\bullet};
(0,-11)*{\mbox{origin}};
(108,-3)*{\mbox{\begin{footnotesize}axis\end{footnotesize}}};
\end{xy}
\]
\caption{If label $k$ is closer to the origin in $\widehat{B}$ than in $B$, then $\widehat{m}_k$ must have been labeled already in $B$ by some $j<k$, i.e. $\widehat{m}_k \mathop{<}_D m_k$.}
  \label{figure:proofaid}
\end{figure}

Note also that if $\widehat{m}_k \mathop{<}_{\widehat{D}} m_k$, then in the crossout sequence for $\widehat{B}$, both $\widehat{m}_k$ and $m_k$ are available at step $k$ and $\widehat{m}_k$ is chosen.  Therefore, the label $k$ is placed closer to the origin in $\widehat{B}$ than in $B$.  Hence
  \begin{equation}
    \widehat{m}_k \mathop{<}_{\widehat{D}} m_k \implies k\mbox{ is jumpy}.
  \label{eqn:kjumpy2}
\end{equation}

Now suppose for a contradiction that one of the labels $1,\ldots,|\widehat{D}|$ is jumpy, and let $k$ be the smallest jumpy label. 
Since $\widehat{m}_k \in \widehat{D}$ and $\widehat{D}\subset D$, the morsel $\widehat{m}_k$ also belongs to $D$.  Let $j$ be its label on the crossout board of $D$; that is,
\begin{equation}
  m_j = \widehat{m}_k.
  \label{eqn:jk}
\end{equation}
Then 
  \begin{align*}
    k \mbox{ is jumpy} &\implies \widehat{m}_k \mathop{<}_D m_k & \mbox{by \eqref{eqn:kjumpy1}}\\
    &\implies m_j \mathop{<}_D m_k & \mbox{by \eqref{eqn:jk}}\\
  &\implies j < k  & \mbox{by \eqref{eqn:jlek}}\\
  &\implies \widehat{m}_j \mathop{<}_{\widehat{D}} \widehat{m}_k & \mbox{by \eqref{eqn:jlek}}\\
  &\implies \widehat{m}_j \mathop{<}_{\widehat{D}} m_j& \mbox{by \eqref{eqn:jk}}\\
  &\implies j \mbox{ is jumpy} & \mbox{by \eqref{eqn:kjumpy2}}
\end{align*}
That is, $j<k$ and $j$ is jumpy.  But $k$ was the smallest jumpy label.  This contradiction shows that there are no jumpy labels, completing the proof.
\end{proof}

The \emph{crossout scores} $\chi_A(D)$ and $\chi_B(D)$ are the scores for Alice and Bob when both play the crossout strategy: 
	\begin{align*}
	\chi_A(D) = v_A^D(c,c) &= m_2 + m_4 + \cdots + m_{2\floor{n/2}},  \\
	\chi_B(D) = v_B^D(c,c) &= m_1 + m_3 + \cdots + m_{2\ceiling{n/2}-1}.
	\end{align*}
These scores are easy to read off from the crossout board. The \emph{unlabeled} morsel locations on a player's axis are precisely the utilities of the morsels he eats if both players follow the crossout strategy.  Therefore, the crossout scores $\chi_A(D)$ and $\chi_B(D)$ are obtained by summing the unlabeled locations (marked with dashes in Figure \ref{figure:crossoutboard}) on the $a$- and $b$-axes respectively.  For instance, for the board $D$ pictured in Figure~\ref{figure:crossoutboard}, we have $\chi_A(D) = 4+5+6+8$ and $\chi_B(D) = 3+5+6+8$.

If we also wish to show the order of play, then we can label the $a$-coordinate of the morsel eaten by Alice in turn $i$ with the symbol $A_i$, and the $b$-coordinate of the morsel eaten by Bob in turn $j$ with the symbol $B_j$ as shown in Figure~\ref{figure:crossoutboardscores}.  Alice's score is the sum of the $a$-coordinates labeled with $A$'s, and Bob's score is the sum of the $b$-coordinates labeled with $B$'s.  

\begin{figure}
\[
\begin{xy}
0;
(0.6,0):
(0,0)*{}; (70,0)*{} **\dir{-};
(0,10)*{}; (70,10)*{} **\dir{-};
(0,20)*{}; (70,20)*{} **\dir{-};
(0,30)*{}; (70,30)*{} **\dir{-};
(0,40)*{}; (70,40)*{} **\dir{-};
(0,50)*{}; (70,50)*{} **\dir{-};
(0,60)*{}; (70,60)*{} **\dir{-};
(0,70)*{}; (70,70)*{} **\dir{-};
(0,0)*{}; (0,70)*{} **\dir{-};
(10,0)*{}; (10,70)*{} **\dir{-};
(20,0)*{}; (20,70)*{} **\dir{-};
(30,0)*{}; (30,70)*{} **\dir{-};
(40,0)*{}; (40,70)*{} **\dir{-};
(50,0)*{}; (50,70)*{} **\dir{-};
(60,0)*{}; (60,70)*{} **\dir{-};
(70,0)*{}; (70,70)*{} **\dir{-};
(30,30)*{\bullet};
(10,20)*{\bullet};
(60,40)*{\bullet};
(50,10)*{\bullet};
(20,50)*{\bullet};
(70,60)*{\bullet};
(0,70)*{\bullet};
(40,0)*{\bullet};
(-10,0)*{2};
(-10,10)*{4};
(-10,20)*{B_6};
(-10,30)*{6};
(-10,40)*{B_2};
(-10,50)*{B_4};
(-10,60)*{8};
(-10,70)*{B_8};
(0,-10)*{1};
(10,-10)*{3};
(20,-10)*{5};
(30,-10)*{A_3};
(40,-10)*{A_7};
(50,-10)*{A_5};
(60,-10)*{7};
(70,-10)*{A_1};
(-22,11)*{\begin{sideways}\begin{footnotesize}$\mbox{Bob ($b$)}\longrightarrow$\end{footnotesize}\end{sideways}};
(15,-23)*{\mbox{\begin{footnotesize}Alice ($a$)\end{footnotesize}}\longrightarrow};
\end{xy}
\]
  \caption{A crossout board showing the sequence of play: Alice eats the morsel above the label $A_i$ on turn $i$, and Bob eats the morsel to the right of the label $B_j$ on turn $j$.}
  \label{figure:crossoutboardscores}
\end{figure}
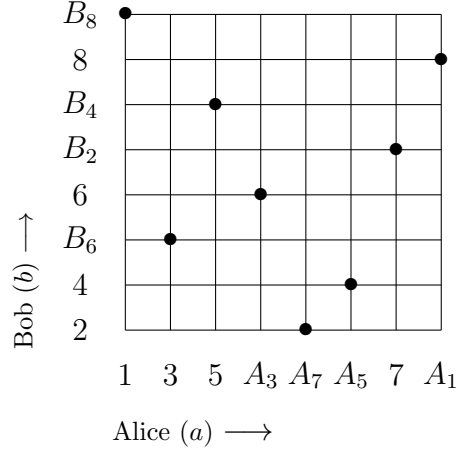

In $D$ of Figure \ref{figure:crossoutboardscores} we see that Alice, who plays first, eats her favourite morsel $(8,7)$ on her first turn.  In the remaining game $D - (8,7)$, Bob moves first but does not eat his favourite morsel $(1,8)$ until his last move (for such is Alice's loathing for it that he can safely ignore it until the end).  An interesting property of the crossout strategy, which we leave as an exercise to the reader since it is not needed for the proof of the main theorem, is that if both players follow it, then the first player eventually eats her favourite morsel. 

\subsection*{The main lemma}
The next lemma shows that neither player can improve his crossout score by choosing a different first morsel.  

\begin{lemma}[Main Lemma]
  \label{lem:mainlemma}
  Let $D$ be a dinner, and let $m$ be a morsel of $D$.  Let $P$ be the player to move first in $D$.  Then
  \begin{equation*}
    u_P(m) + \chi_P(D-m) \le \chi_P(D).
  \end{equation*}
\end{lemma}

\begin{proof}
  We compare the crossout boards for $D$ and $D-m$ (Figure~\ref{figure:crossoutboard}).  In the latter, a morsel has been removed.  Player $P$ is the second player to move in $D-m$, so he swallows one fewer morsel in $D-m$ than in $D$.  This means the number of labels on $P$'s axis is the same in the crossout boards of $D$ and $D-m$.  By the Crossout Board Lemma~\ref{lem:nmoveout}, each label on the board for $D - m$ is no closer to the origin than the corresponding label on the board for $D$.  Therefore the sum of the labeled positions on $P$'s axis is at least as large in $D-m$ as in $D$.  Hence the sum of the \emph{unlabeled} positions on $P$'s axis is no larger in $D-m$ than in $D$.  For $D$, this sum is the crossout score $\chi_P(D)$.  For the board $D-m$, this sum consists of the score $\chi_P(D-m)$ plus the utility $u_P(m)$ of the removed morsel~$m$.
 \end{proof}
 
\subsection*{Proof of Theorem \ref{thm:crossoutequil}}


Let $D$ be a dinner of $n$ morsels, and let $c$ be the crossout strategy.  We induct on $n$ to show that for any player $P \in \{A,B\}$ and any strategy $s$,
	\begin{align*}
	 v_P^D(s,c) \leq v_P^D(c,c).
	 \end{align*}
The base case $n=1$ is trivial because $c$ is the only strategy: In a game with one morsel, the only thing you can do is eat it!
	
On to the inductive step. Suppose first that $P$ is the first player to move in $D$.  Let $m = s(D)$.  Then	 
 \begin{align*}
   v_P^D(s,c) &= u_P(m) + v_P^{D-m}(s,c) & \mbox{by \eqref{eq:scoredef}} \\
    		&\leq u_P(m) + v_P^{D-m}(c,c) & \mbox{by the inductive hypothesis} \\
      &\le v_P^D(c,c) & \mbox{\hfill by Main Lemma \ref{lem:mainlemma}.}
  \end{align*}
It remains to consider the case that $P$ is the second player to move in $D$.  Letting $m=c(D)$, we have by the inductive hypothesis and \eqref{eq:scoredef}, 
	\[ v_P^D(s,c) = v_P^{D-m}(s,c) \leq v_P^{D-m}(c,c) = v_P^D(c,c) \]
which completes the proof.

\section*{Concluding Remarks}

We have analyzed Ethiopian Dinner as a non-cooperative game, and found an efficiently computable subgame perfect equilibrium, the crossout strategy.  Here we discuss its efficiency, and mention some variants and generalizations.

\subsection*{Pareto efficiency and inefficiency}

An \emph{outcome} of a dinner $D$ is a partition of its morsels between Alice and Bob.
An outcome is called \emph{Pareto inefficient} if there exists another outcome that is at least as good for both players and is strictly better for one of them.  
Equilibrium strategies may result in an outcome that is Pareto inefficient, as demonstrated by the famous Prisoner's Dilemma, in which both players do better by mutual cooperation than by mutual defection even though mutual defection is the unique equilibrium~\cite{Owen}.  


For the permutation dinner shown in Figure~\ref{fig:typicalscores}, we see that the crossout outcome $(c,c)$ is Pareto efficient because there are no dots lying (weakly) both above and to the right of the crossout score ({\Large ${\color{red}\bullet}$}).  In fact, Brams and Straffin \cite[Theorem~1]{BS} prove that the crossout outcome is always Pareto efficient with respect to the ``pairwise comparison'' partial order on outcomes.  On the other hand, the crossout outcome is \emph{not} always Pareto efficient for the score function we have been considering, the sum of the utilities of morsels eaten.
  Among permutation dinners, the smallest Pareto inefficient examples occur for dinners of size~$6$, for which there are two:
\begin{align*}
D_1 &= \{(1, 5), (2, 1), (3, 2), (4, 4), (5, 6), (6, 3) \}, \\
D_2 &= \{(1, 5), (2, 1), (3, 2), (4, 3), (5, 4), (6, 6) \}.
\end{align*}
If each player employs his greedy strategy in dinner $D_1$, then Alice's score of $6+4+3$ is equal to her crossout score, while Bob's score of $6+5+1$ is one point better than his crossout score.
The outcome that improves on crossout for $D_2$ is rather interesting: Bob can do one point better than his crossout score of $11$ without hurting Alice's score, but only if Alice magnanimously leaves him the most delicious morsel $(6,6)$ instead of gobbling it up on the first bite!  Bob must then return the favor by refraining from eating $(3,2)$ on the fourth move.  Because of the trust it requires, it is hard to imagine the players achieving this outcome if they are not able to communicate.

To see how often crossout is Pareto inefficient, we used Sage Mathematics Software~\cite{Sage} to check 10\,000 randomly generated permutation dinners of size~16 for Pareto efficiency:  In all but 672 cases the $(c,c)$ outcome was Pareto efficient, and in all but 241 it was weakly Pareto efficient (that is, no other outcome resulted in strict improvements for both players).  The improvement in scores achieved by alternate outcomes was small: in an exhaustive check of all possible game outcomes in all $10\,000$ dinners, the largest improvement for any player was less than $8 \%$.  These findings provide some evidence that the crossout strategy is reasonably efficient.  


\subsection*{Generalized payoffs}

The outcome of an Ethiopian Dinner is a partition of the index set $\{1,\ldots,n\}$ into a set $\mathcal{A}=\{i_1,\ldots,i_{\floor{n/2}}\}$ of $\floor{n/2}$ morsels eaten by Alice and a set $\mathcal{B} = \{j_1,\ldots,j_{\ceiling{n/2}}\}$ of $\ceiling{n/2}$ morsels eaten by Bob.  We have assumed that the final scores (payoffs) for Alice and Bob  take the form
	\[ p_A = \sum_{i \in \mathcal{A}} a_i \qquad \mbox{and} \qquad p_B = \sum_{j \in \mathcal{B}} b_j. \]
This particular payoff function is not essential for the argument, however.  Let
	\[ f_A : \RR^{\floor{n/2}} \to \RR \qquad \mbox{and} \qquad f_B : \RR^{\ceiling{n/2}} \to \RR \]  
be functions that are strictly increasing in each coordinate, and symmetric with respect to permutations of the coordinates.  Then the Ethiopian Dinner game with payoffs 
	\[ p_A = f_A(a_{i_1},\ldots,a_{i_{\floor{n/2}}}) \qquad \mbox{and} \qquad p_B = f_B(b_{j_1},\ldots,b_{j_{\ceiling{n/2}}}) \]
has crossout as its optimal strategy.  Indeed, the proof we have given uses only the relative order of the $a_i$ and the $b_j$, and not their actual values. 

One could also generalize the payoff function so that Alice's payoff depends not only on the morsels she ate but also on the morsels Bob ate, and vice versa.  A natural choice is
	\begin{align*} p_A &= \alpha_A \sum_{i \in \mathcal{A}} a_i + \beta_A \sum_{j \in \mathcal{B}} b_j, \qquad \mbox{and} \\
	 p_B &= \alpha_B \sum_{i \in \mathcal{A}} a_i + \beta_B \sum_{j \in \mathcal{B}} b_j. \end{align*}
That is, Alice's payoff is $\alpha_A$ times the sum of her own utilities of the morsels she ate, plus $\beta_A$ times the sum of the utilities to Bob of the morsels Bob ate.  Bob's payoff is defined similarly.  
The ratios $\beta_A/ \alpha_A$ and $\alpha_B / \beta_B$ measure the degree of altruism of the two players.    The scenario of friends eating in an Ethiopian restaurant might correspond to values of these ratios strictly between $0$ and $1$.  One can also imagine scenarios where $\beta_A / \alpha_A >1$: perhaps Alice is Bob's mother and the morsels in question are brussels sprouts.

All of these games turn out to be equivalent to Ethiopian Dinner.  Suppose we are considering the above payoffs $p_A$ and $p_B$ on the dinner $D$ consisting of morsels $m_i = (a_i, b_i)$ for $i=1,\ldots,n$.  Translating all of a player's utilities by an additive constant has no effect on strategy, so we may assume that $\sum_{i=1}^n a_i = \sum_{i=1}^n b_i = 0$.  Then
\[
\sum_{j \in \mathcal{B}} b_j = - \sum_{i \in \mathcal{A}} b_i \qquad \mbox{and} \qquad \sum_{i \in \mathcal{A}} a_i = - \sum_{j \in \mathcal{B}} a_j.
\]
Now let $D'$ be the dinner consisting of morsels
\[
m_i' = (\alpha_A a_i - \beta_A b_i,\, \beta_B b_i - \alpha_B a_i ), \qquad \mbox{ for } i=1,\ldots,n.
\]
Any strategy $s$ on $D$ has a corresponding strategy $s'$ on $D'$ (which chooses $m'_i$ whenever $s$ chooses $m_i$), and
\[
 p_P^D(s,t) = v_P^{D'}(s',t')
\]
for both players $P \in \{A,B\}$.
In other words, the modified payoff in~$D$ equals the usual Ethiopian Dinner payoff in~$D'$.  Now let $s$ be the strategy on $D$ for which $s'=c$ (that is, $s'$ is the crossout strategy on $D'$).  Then the pair $(s,s)$ is an equilibrium for the modified payoff dinner~$D$.

We distinguish two extreme cases.

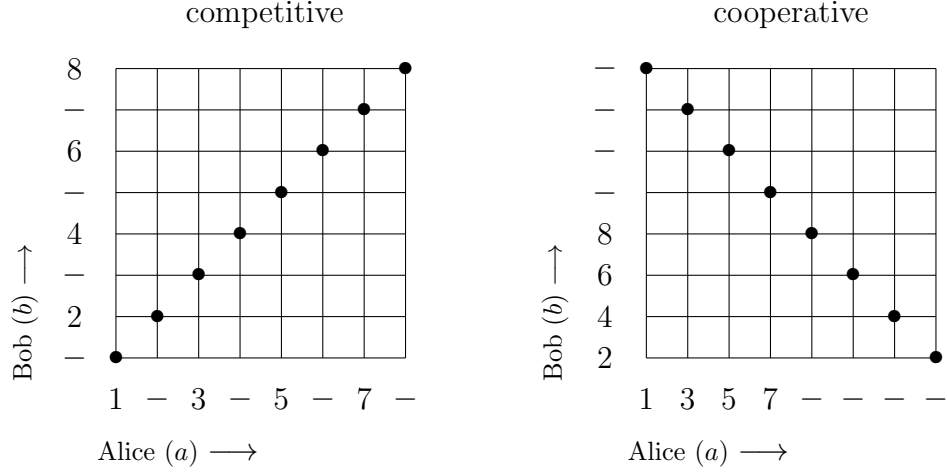
\begin{figure}
\[
\begin{xy}
0;
(0.55,0):
(0,0)*{}; (70,0)*{} **\dir{-};
(0,10)*{}; (70,10)*{} **\dir{-};
(0,20)*{}; (70,20)*{} **\dir{-};
(0,30)*{}; (70,30)*{} **\dir{-};
(0,40)*{}; (70,40)*{} **\dir{-};
(0,50)*{}; (70,50)*{} **\dir{-};
(0,60)*{}; (70,60)*{} **\dir{-};
(0,70)*{}; (70,70)*{} **\dir{-};
(0,0)*{}; (0,70)*{} **\dir{-};
(10,0)*{}; (10,70)*{} **\dir{-};
(20,0)*{}; (20,70)*{} **\dir{-};
(30,0)*{}; (30,70)*{} **\dir{-};
(40,0)*{}; (40,70)*{} **\dir{-};
(50,0)*{}; (50,70)*{} **\dir{-};
(60,0)*{}; (60,70)*{} **\dir{-};
(70,0)*{}; (70,70)*{} **\dir{-};
(0,0)*{\bullet};
(10,10)*{\bullet};
(20,20)*{\bullet};
(30,30)*{\bullet};
(40,40)*{\bullet};
(50,50)*{\bullet};
(60,60)*{\bullet};
(70,70)*{\bullet};
(-10,0)*{-};
(-10,10)*{2};
(-10,20)*{-};
(-10,30)*{4};
(-10,40)*{-};
(-10,50)*{6};
(-10,60)*{-};
(-10,70)*{8};
(0,-10)*{1};
(10,-10)*{-};
(20,-10)*{3};
(30,-10)*{-};
(40,-10)*{5};
(50,-10)*{-};
(60,-10)*{7};
(70,-10)*{-};
(-22,11)*{\begin{sideways}\begin{footnotesize}$\mbox{Bob ($b$)}\longrightarrow$\end{footnotesize}\end{sideways}};
(15,-23)*{\mbox{\begin{footnotesize}Alice ($a$)\end{footnotesize}}\longrightarrow};
  (36,83)*{\mbox{competitive}}
\end{xy}
\quad\quad\quad\quad
\begin{xy}
0;
(0.55,0):
(0,0)*{}; (70,0)*{} **\dir{-};
(0,10)*{}; (70,10)*{} **\dir{-};
(0,20)*{}; (70,20)*{} **\dir{-};
(0,30)*{}; (70,30)*{} **\dir{-};
(0,40)*{}; (70,40)*{} **\dir{-};
(0,50)*{}; (70,50)*{} **\dir{-};
(0,60)*{}; (70,60)*{} **\dir{-};
(0,70)*{}; (70,70)*{} **\dir{-};
(0,0)*{}; (0,70)*{} **\dir{-};
(10,0)*{}; (10,70)*{} **\dir{-};
(20,0)*{}; (20,70)*{} **\dir{-};
(30,0)*{}; (30,70)*{} **\dir{-};
(40,0)*{}; (40,70)*{} **\dir{-};
(50,0)*{}; (50,70)*{} **\dir{-};
(60,0)*{}; (60,70)*{} **\dir{-};
(70,0)*{}; (70,70)*{} **\dir{-};
(0,70)*{\bullet};
(10,60)*{\bullet};
(20,50)*{\bullet};
(30,40)*{\bullet};
(40,30)*{\bullet};
(50,20)*{\bullet};
(60,10)*{\bullet};
(70,0)*{\bullet};
(-10,0)*{2};
(-10,10)*{4};
(-10,20)*{6};
(-10,30)*{8};
(-10,40)*{-};
(-10,50)*{-};
(-10,60)*{-};
(-10,70)*{-};
(0,-10)*{1};
(10,-10)*{3};
(20,-10)*{5};
(30,-10)*{7};
(40,-10)*{-};
(50,-10)*{-};
(60,-10)*{-};
(70,-10)*{-};
(-22,11)*{\begin{sideways}\begin{footnotesize}$\mbox{Bob ($b$)}\longrightarrow$\end{footnotesize}\end{sideways}};
(15,-23)*{\mbox{\begin{footnotesize}Alice ($a$)\end{footnotesize}}\longrightarrow};
  (35,83)*{\mbox{cooperative}}
\end{xy}
\]
 \caption{Crossout boards for a fully competitive and fully cooperative dinner.}
  \label{fig:competitiveness}
\end{figure}

If $\alpha_A = \beta_B = 1$ and $\alpha_B = \beta_A =-1$, then the game is zero-sum.  In the terminology of combinatorial game theory, each morsel $m=(a,b)$ is a \emph{switch} $\{a|-b\}$, so the full game is a sum of switches.  The morsel $(a,b)$ has temperature $a+b$, and optimal play proceeds in order of decreasing temperature  (see \cite{BCG} for background).  The equivalent Ethiopian Dinner $D'$ has morsels $(a_i+b_i,a_i+b_i)$ of equal appeal to both players, and crossout on $D'$ gives the same decreasing-temperature play.  This dinner could result in a couple of burned tongues!

If $\alpha_A = \alpha_B = \beta_A = \beta_B =1$, then the game is fully cooperative.  Both players have the same goal of maximizing their joint welfare.  Since the game rules constrain them to alternate moves, the optimal play is the following.  Order the morsels $m_1,\ldots,m_n$ so that $a_i - b_i$ is a decreasing function of $i$.  Alice takes morsels $m_1,\ldots,m_{\floor{n/2}}$, and Bob takes morsels $m_{\floor{n/2}+1},\ldots,m_n$.  In this case, the equivalent Ethiopian Dinner $D'$ has morsels $(a_i-b_i,b_i-a_i)$ and crossout on $D'$ gives the optimal strategy just described.  Figure~\ref{fig:competitiveness} shows examples of crossout boards for a zero-sum (competitive) dinner and a cooperative dinner.

\subsection*{Combinatorics at the dinner table}

We can measure the ``cooperativeness'' of a permutation dinner by its inversions.  Let $\pi=(\pi_1,\ldots,\pi_n)$ be a permutation of $1,\ldots,n$.  For each pair of indices $i<j$ such that $\pi_i >\pi_j$, we call the pair $(i,j)$ a \emph{left inversion} of $\pi$ and the pair $(\pi_i,\pi_j)$ a \emph{right inversion} of $\pi$.  Both players should be pleased with a permutation dinner if it has a lot of inversions, because each inversion represents a pair of morsels $m_i$, $m_j$ such that Alice prefers $m_j$ while Bob prefers $m_i$.  Hopkins and Jones~\cite{HJ} show that if the left inversions of~$\pi$ are a subset of the left inversions of~$\pi'$, then Alice's crossout score for the permutation dinner $\pi'$ is at least as good as for~$\pi$.  In fact they show more: there is a bijection between the set of morsels Alice eats in~$\pi$ and the set she eats in~$\pi'$ such that each morsel eaten in~$\pi'$ is at least as tasty to Alice as the corresponding one in~$\pi$. (Alice prefers a prime piece of pie to an ordinary one: after all, who wouldn't?) Likewise, Bob prefers dinners with a lot of right inversions. (Curiously, although right inversions are in bijection with left inversions, set inclusion of right inversions induces a different partial ordering on permutations than does set inclusion of left inversions, as the reader can verify for permutations of $3$ elements!)

Consider permutation dinners having an even number of morsels, $n=2k$.  If we define an \emph{outcome} for Alice as the set of utilities of the morsels she eats, then the number of conceivable outcomes (over all permutation dinners of size $2k$) is the binomial coefficient $2k \choose k$.  But not all of these outcomes are achievable if both players play the crossout strategy.  For instance, Alice never has to eat her least favourite morsel (of utility $1$), and she eats at most one of her three least favourite morsels (of utilities $1$, $2$ and $3$).  In general, for each $j\geq 0$ she eats at most $j$ of her $2j+1$ least favourite morsels.  Hopkins and Jones \cite{Hopkins, HJ} show that if both players play crossout, then the number of attainable outcomes for Alice is the famous Catalan number $C_k = \frac{1}{k+1} {2k \choose k}$, and the number of attainable outcomes for Bob is $C_{k+1}$.

\subsection*{Cake cutting and envy-free division}

There is a large literature on \emph{cake-cutting} \cite{BT} in which a cake (identified with the interval $[0, 1]$) comes equipped with a measure for each player describing the utility to him of eating a given piece. One problem is to find an \emph{envy-free} partition of the cake, which means that each player prefers the piece assigned to him over the pieces assigned to the other players. When the cake is comprised of indivisible
slices (as, for example, in a game of Ethiopian Dinner), this criterion becomes impossible to achieve in general, and
finding an envy-minimizing allocation
is a hard computational problem~\cite{LMMS}.  The outcome of the crossout
strategy is reasonably close to envy-free: the first player is not
envious, and the second player's envy is bounded by the utility of his
favourite morsel.  Another way of achieving an approximately envy-free
allocation is described in \cite[Theorem 2.1]{LMMS}.  The algorithm described
there is even faster than crossout, because it does not require
sorting the lists of utilities.  In the case when an envy-free allocation exists, the \emph{undercut procedure} of~\cite{BKK} gives an algorithm for finding one.

\subsection*{First move advantage}

Who takes the first bite at dinner might seem of little consequence, but imagine instead that Alice and Bob are taking turns dividing up the assets of their great uncle's estate.  Alice, choosing first, claims the mansion, leaving Bob to console himself with the sports car.  To reduce the advantage of moving first, Brams and Taylor \cite[Ch.\ 3]{BT2} propose what they call ``balanced alternation,'' which uses the Thue-Morse sequence
	\[ A,B,B,A,B,A,A,B,B,A,A,B,A,B,B,A,\ldots \]
as the order of turns.  Alice chooses the first item, then Bob chooses the second and third, Alice the fourth, and so on.  The Thue-Morse sequence is the unique sequence starting with $A$ that is fixed under the operation of replacing each term $A$ by $A,B$ and each term $B$ by $B,A$.  Here is another definition that the reader might enjoy proving is equivalent: the terms are indexed starting from $n=0$, and the $n$-th term equals $A$ or $B$ according to whether the sum of the binary digits of $n$ is even or odd.  Many amazing properties of this sequence can be found in~\cite{AS}.

A further innovation of \cite{BT2} aimed at reducing the first move advantage is simultaneous choices.  On each turn both players choose a morsel, and if they choose distinct morsels then each eats his choice.  If they choose the same morsel, then that morsel is marked as ``contested'' and simultaneous play continues with the remaining uncontested morsels.  When only contested morsels remain, the players revert to the Thue-Morse order of play.  

It would be interesting to quantify the intuition that the Thue-Morse order tends to produce a fair outcome.  Some assumptions on the utilities would be needed: In the example of the estate, if the value of the mansion exceeds the combined value of all the other assets, then the first player retains an advantage no matter what.
 
\subsection*{Inducing sincerity}

\emph{Mechanism design} asks how the rules of a game should be structured in order to induce the players to act in a certain way.  A common goal is to induce the players to make sincere choices -- that is, choices in line with their actual preferences.  For example, the voter who deserts his preferred candidate for one he deems has a better chance of winning is making an \emph{in}sincere choice.  Making sincere choices in Ethiopian Dinner means playing the greedy strategy of always eating your favourite available morsel.  Because the greedy strategy pair is not an equilibrium, the players have an incentive to make insincere choices by playing the more complicated crossout strategy.  How could the rules of Ethiopian Dinner be modified so as to make the greedy strategy pair an equilibrium?  Brams and Kaplan~\cite{BK} (see also \cite[ch.\ 9]{Brams}) design a mechanism that allows the players of an Ethiopian Dinner to offer one another trades, and show that its effect is to make the greedy strategy pair an equilibrium.



\section*{Open questions}

We conclude by describing a few natural variants that we do not know how to analyze.

\subsection*{Delayed gratification} 

Suppose that the utilities $u_A(m)$ and $u_B(m)$ depend not only on the morsel $m$, but also on when it is eaten.  A natural choice is to value a morsel eaten on turn~$i$ with $\lambda^i$ times its usual value, for a parameter $\lambda < 1$.  Thus, a morsel declines in value the longer it remains on the plate (perhaps the delicate flavours are fading).  The choice of exponential decay $\lambda^i$ corresponds to the common assumption in economics that a payoff received in the future should be discounted to its net present value according to the prevailing interest rate.  If the interest rate is $\alpha$, then $\lambda = 1/(1+\alpha)$.  In the resulting game, each player feels an urgency to eat her favourites early on.  Because time-sensitive payoffs break the symmetry assumption, our proof of equilibrium does not apply.  Can the crossout strategy be modified to produce an equilibrium?

\subsection*{Inaccessible morsels}
Ethiopian food is served atop \emph{injera}, a layer of spongy bread that can only be eaten once it is revealed.  If the game is played with the requirement that the order of consumption must respect a fixed partial ordering on the morsels, the crossout strategy may not be an allowable strategy.  What should take its place?

\subsection*{Imperfect information}

How should your dinner strategy change if you don't know your opponent's preferences?  Perhaps the morsels come in different types and there are many identical morsels of each type.  Then you might start by playing greedily, while trying to infer your opponent's preferences from his play.  Once you start to form a picture of his preferences, you can start playing a crossout-style strategy.  Bob seems to be staying away from the spinach, so you can probably leave it for later.  But was it an honest signal when he ate all those chickpeas?  Can deceit pay in a dinner of imperfect information?

\subsection*{Three's a crowd}

Ethiopian Dinner resembles the process of draft picks in sports.  Each team participating in the draft has its own belief about how much each player is worth, and the teams draft players one at a time according to some predetermined order of play.  Typically, many teams (more than two!) participate in the draft.  Brams and Straffin~\cite{BS} point out a number of pathologies in the case when the number of teams is greater than two.  For example, it may be to a team's advantage to choose later in the draft.   For players with certain specific preferences (called ``gallant knights'' in \cite{LSS}) the crossout strategy applies to games with any number of players.  But the crossout strategy does not seem to apply to games of three or more self-interested players, which leads us to end with a question. Is there an efficiently computable equilibrium for Ethiopian Dinner with three or more players? 

\begin{acknowledgement}
The first author's research was partially supported by an NSF Postdoctoral Research Fellowship.  The second author's research was supported by NSERC PDF-373333 and NSF MSPRF 0802915, and was partially performed while the author was at Simon Fraser University and the Pacific Insitute for the Mathematical Sciences at the University of British Columbia.  We thank the Centre for Experimental and Constructive Mathematics at Simon Fraser University for providing computer resources, and we are grateful to Elwyn Berlekamp, Kevin Doerksen,  Eric LeGresley, James Propp, Scott Sheffield and Jonathan Wise for useful suggestions.  Thanks to Renato Paes Leme for bringing reference \cite{KC} to our attention, and to anonymous referees whose astute suggestions significantly improved the paper.

We especially thank Asmara Restaurant in Cambridge, MA for the shared meal that inspired this paper.
(The clear winner on that occasion was the first author's wife, whose strategy was to ignore all talk of game theory and eat all of her favourite morsels while her dinner companions were distracted by analyzing the game!)
\end{acknowledgement}


\begin{thebibliography}{10}

\bibitem{AS} J.-P. Allouche and J. Shallit, The ubiquitous Prouhet-Thue-Morse sequence, in
\emph{Sequences and Their Applications: Proceedings of SETA '98}, C. Ding, T. Helleseth, and H. Niederreiter eds., pages 1--16, 1999.
\bibitem{BCG} E. Berlekamp, J. Conway, and R. Guy, \emph{Winning ways for your mathematical plays}, Vol.\ 1, 2nd ed., A. K. Peters/CRC Press, 2001.
\bibitem{Brams} S. J. Brams, \emph{Mathematics and Democracy: Designing Better Voting and Fair Division Procedures}, Princeton, 2008.
\bibitem{BK} S. J. Brams and T. R. Kaplan, Dividing the indivisible: procedures for allocating cabinet ministries in a parliamentary system, \emph{J. Theor.\ Politics} \textbf{16}(2):143--173, 2004.
\bibitem{BKK} S. J. Brams, D. M. Kilgour and C. Klamler, The undercut procedure: an algorithm for the envy-free division of indivisible items, \emph{Social Choice and Welfare}, to appear.
\bibitem{BS} S. J. Brams and P. D. Straffin, Jr., Prisoners' dilemma and
professional sports drafts, \emph{Amer.\ Math.\ Monthly} \textbf{86}(2):80--88, 1979.
\bibitem{BT} S. J. Brams and A. D. Taylor, \emph{Fair division: from cake-cutting to dispute resolution}, Cambridge University Press, 1996.
\bibitem{BT2} S. J. Brams and A. D. Taylor, \emph{The win-win solution}, Norton, 1999.
\bibitem{DGP} C. Daskalakis, P. W. Goldberg, and C.H. Papadimitriou. The complexity of computing a
Nash equilibrium. In \emph{Proceedings of the 38th annual ACM symposium on theory of
computing} (STOC), pages 71--78, 2006.
\bibitem{Hopkins} B. Hopkins, Taking turns, \emph{College Math.\ J.} \textbf{41}(4):289--297, 2010.
\bibitem{HJ} B. Hopkins and M. A. Jones, Bruhat orders and the sequential selection of indivisible items, in \emph{The mathematics of preference, choice and order}, pages 273--285, Springer, 2009.
\bibitem{KC} D. A. Kohler and R. Chandrasekaran, A class of sequential games, 
\emph{Operations Research} \textbf{19}: 270--277, 1971.
\bibitem{LSS} L. Levine, S. Sheffield and K. E. Stange, A duality principle for selection games (to appear).
\bibitem{LMMS} R. Lipton, E. Markakis, E. Mossel and A. Saberi, On approximately fair allocations of indivisible goods. In \emph{EC '04: Proceedings of the 5th ACM conference on electronic commerce}, pages 125--131, 2004.
\bibitem{Owen} G. Owen, \emph{Game theory}, 3rd ed., Academic Press, 1995.
\bibitem{Rasm} E. Rasmusen, \emph{Games and information:  an introduction to game theory}, 4th ed., Blackwell Publishers, 2006.
\bibitem{Sage} W. A. Stein and others, {S}age {M}athematics {S}oftware ({V}ersion 4.6.1), 2011. \newblock \url{http://www.sagemath.org/}. To experiment with Ethiopian Dinner on a computer and generate plots like Figure~\ref{fig:typicalscores}, visit \url{http://math.katestange.net/scripts.html}.

\end{thebibliography}
\end{document}